\documentclass[amsmath,amsfonts,aps,onecolumn,twoside,superscriptaddress]{revtex4}

\usepackage{amsfonts}
\usepackage{amsmath}
\usepackage{graphics}
\usepackage{graphicx}
\usepackage{enumerate}
\usepackage{amssymb}
\usepackage{amsthm}
\usepackage{subfigure}
\newcommand{\appref}[1]{\hyperref[#1]{{Appendix~\ref*{#1}}}}
\newcommand{\be}{\begin{eqnarray} \begin{aligned}}
\newcommand{\ee}{\end{aligned} \end{eqnarray} }
\newcommand{\benn}{\begin{eqnarray*} \begin{aligned}}
\newcommand{\eenn}{\end{aligned} \end{eqnarray*}}
\newcommand*{\cA}{\mathcal{A}} 
\newcommand*{\cB}{\mathcal{B}}

\newcommand*{\cE}{\mathcal{E}}

\newcommand*{\cL}{\mathcal{L}}

\newcommand*{\cN}{\mathcal{N}}

\newcommand{\bc}{\begin{center}}
\newcommand{\ec}{\end{center}}

\newcommand{\Tr}{\mathop{\mathrm{tr}}\nolimits}

\newtheorem{theorem}{Theorem}[section]
\newtheorem{lemma}[theorem]{Lemma}

\usepackage{amsfonts}

\def\01{\{0,1\}}

\begin{document}

\title{Limits on classical communication from quantum entropy power inequalities}

\author{Robert K\"onig}
\affiliation{IBM TJ Watson Research Center 1101 Kitchawan Road, Yorktown Heights, NY 10598} 
\author{Graeme Smith}
\affiliation{IBM TJ Watson Research Center 1101 Kitchawan Road, Yorktown Heights, NY 10598} 

\date{\today}

\begin{abstract}
Almost all modern communication systems rely on electromagnetic fields as a means of information transmission, and 
finding the capacities of these systems is a problem of significant practical importance.  The Additive White Gaussian Noise (AWGN)
channel is often a good approximate description of such systems, and its capacity is given by a simple formula.  
However, when quantum effects are important, estimating the 
capacity becomes difficult: a lower bound is known, but a similar upper bound is missing.  
Here we present strong new upper bounds
for the classical capacity of quantum additive noise channels, including quantum analogues of the AWGN channel.  
Our main technical tool is a quantum
entropy power inequality  that controls the entropy production as two quantum 
signals combine at a beam splitter.  Its 
proof involves a new connection between entropy production rates and a quantum Fisher information, 
and uses a quantum diffusion that smooths arbitrary states towards gaussians. 
\end{abstract}

\maketitle

\section{Introduction and statement of results}

Channel capacity is central to Shannon's information theory \cite{Shannon48}.  Operationally, it is the 
maximum achievable communication rate, measured in bits per channel use.  Mathematically, it 
is the maximum correlation that can be generated with a single use of the channel, with correlation measured by the mutual 
information.  Practically, it is an optimal measure against which to compare the performance of real systems.

Understanding the impact of quantum effects on channel capacity has been  an important question 
since the early days of information theory \cite{Pierce73}.  The HSW theorem shows 
that the Holevo information, $\chi(\cN) = \max_{\{p_x,\phi_x\}} S(\cN(\bar{\phi}))- \sum_x p_x S(\cN(\phi_x))$, is 
a lower bound for the classical capacity of a quantum channel \cite{SW97,Holevo98}.  
Here $\bar{\phi} = \sum_x p_x \phi_x$ is the average signal state and  $S(\rho) = -\Tr(\rho{\log_2} \rho)$ is
the von Neuman entropy of $\rho$.  The Holevo information is the maximum rate of a code that uses unentangled signal states.  However, using
entangled states, it is sometimes possible to exceed $\chi$ \cite{Hastings09}.

The thermal noise channel, $\cE_{\lambda, N_E}$, 
is a natural quantum analogue for the AWGN channel \cite{Shannon48}, and as such is a good description 
of many practical systems ( Fig. \ref{Fig:ThermalDefinition}).  The HSW lower bound for the capacity of a thermal noise channel with 
average signal photon number $N$ is 
\begin{align}\label{Eq:Holevo}
C(\cE_{\lambda,N_E},N) &\geq \left[ g\left(\lambda N + (1-\lambda)N_E\right)- g\left((1-\lambda)N_E\right)\right]/\ln 2 ,
\end{align}
where $g(x) = (x+1)\ln(x+1) - x\ln x$ \cite{HW01,EW05}.  This communication rate is achievable with a simple classical modulation 
scheme of displaced coherent states \cite{GGLMSY04}, and  exceeding it would require entangled modulation schemes. Our 
goal is to explore the 
usefulness of such novel, fundamentally quantum, strategies.  We find tight bounds on any possible strategy for exceeding Eq.~(\ref{Eq:Holevo})
for a wide range of parameters ( Fig. \ref{Fig:Plots}).  We show such strategies are essentially useless for $\lambda = \frac{1}{2}$ .
Overall, we find for a wide range of practical channels that 
good old classical modulation of coherent states can't be substantially improved upon with quantum tricks.

\section{Entropy Power Inequality}
Entropy is a central quantity in discussions of capacity.
Understanding its properties will be crucial.  For a real variable $X$ with probability density
$p(x)$, the entropy is $H(X)= -\int p(x){\log_2} p(x) dx$.  $H(X)$ measures the information
contained in $X$ and appears in Shannon's formula for the capacity of a noisy channel.  The entropy power,
 $ \frac{1}{2\pi e}e^{2H(X)}$, was considered by Shannon in the context of additive noise 
channels \footnote{An additive noise channel adds independent noise $Y$ to input signal $X$, resulting in output signal $X+Y$.}.  
He proposed that the Entropy Power Inequality (EPI), 

\begin{align}\label{Eq:ShannonEPI}
e^{2H(X+Y)} \geq e^{2H(X)} + e^{2H(Y)},
\end{align}
controlled the entropy production as two statistically independent signals are combined.  Shannon's arguments were incomplete, 
but a full proof of the EPI was given by Stam \cite{Stam59} and Blachman \cite{Blachman65}.  Generalizations of 
the EPI have been found, and recently there
has been renewed interest in streamlining their proofs \cite{VG06,rioultwo}. 
EPIs are a fundamental tool in information theory, 
crucial for bounding capacities of noisy channels in various scenarios \cite{Shannon48,Bergmans74,Cheong78}.  While Eq.~(\ref{Eq:ShannonEPI}) is the most 
commonly cited form, there are several equivalent statements \cite{Demboetal91}.  The following formulation will be most convenient:
\begin{align}\label{Eq:LiebEPI}
H(\sqrt{\lambda}X + \sqrt{1-\lambda}Y) \geq \lambda H(X) + (1-\lambda)H(Y) \ \  {\rm for } \ \ \lambda \in [0,1]. 
\end{align}

A single mode of an electromagnetic field can be described in 
terms of its field quadratures, $P$ and $Q$. When independent modes X and Y with quadratures $(Q_X,P_X)$ and $(Q_Y,P_Y)$ 
are combined at a beam splitter of transmissivity $\lambda$ ( Fig. \ref{Fig:BS}), the signal in one output mode is given
by $(\sqrt{\lambda}Q_X + \sqrt{1-\lambda}Q_Y,\sqrt{\lambda}P_X + \sqrt{1-\lambda}P_Y)$, a process which we denote $X \boxplus_{\lambda}Y$.
Our main result is a quantum analogue of Eq.~(\ref{Eq:LiebEPI}) adapted to this setting, namely, 

\begin{align}\label{Eq:LiebQEPI}
S(X \boxplus_{\lambda} Y) \geq \lambda S(X) + (1-\lambda)S(Y),
\end{align} 
for any independent states on $X$ and $Y$.  This inequality applies unchanged when $X$ and $Y$ are $n$-mode systems.
Here $S(X) = -\Tr \rho_X {\log_2} \rho_X$ is the von Neuman entropy of the state of system $X$, $\rho_X$, with $S(Y)$ 
and $S(X \boxplus_{\lambda} Y)$ defined similarly.  While Eq.~(\ref{Eq:ShannonEPI}) and Eq.~(\ref{Eq:LiebEPI}) are classically equivalent, 
the analogous quantum inequalities do not seem to be.  So, in addition to Eq.~(\ref{Eq:LiebQEPI}), we also prove a quantum analogue
of Eq.(\ref{Eq:ShannonEPI}), valid for beam splitters of transmissivity $1/2$: 
\begin{align}\label{Eq:ShannonQEPI}
e^{\frac{1}{n}S(X \boxplus_{1/2}Y)} \geq \frac{1}{2}e^{\frac{1}{n}S(X)}+ \frac{1}{2}e^{\frac{1}{n}S(Y)}.
\end{align}
 Below we outline a proof of Eq.~(\ref{Eq:LiebQEPI}) and Eq.~(\ref{Eq:ShannonQEPI}), and explore their implications for the classical 
capacity of additive quantum channels.

\begin{figure}[htbp]
\includegraphics[width=3in]{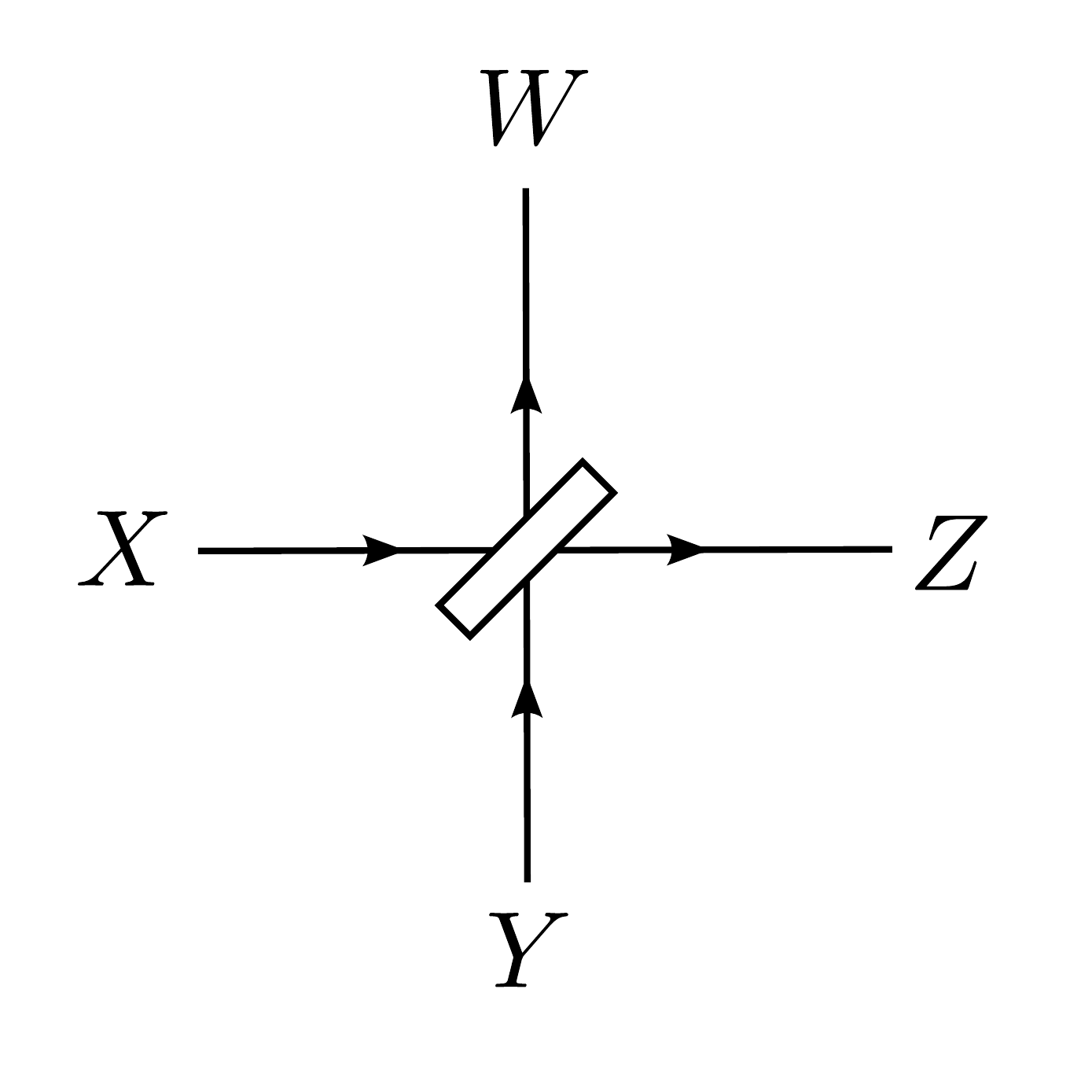}
\caption{Two independent quantum signals combined at a beam splitter.  Both $X$ and $Y$ are $n$-mode systems with quadratures  
$ {R}^X = (Q_1^X, P_1^X, \dots, Q_n^X,P_n^X)$ and $ {R}^Y = (Q_1^Y, P_1^Y, \dots, Q_n^Y,P_n^Y)$.  The output $Z$, which we
denote $X\boxplus_{\lambda}Y$, has quadratures $ R^Z = \sqrt{\lambda} {R}^X +\sqrt{1-\lambda} {R}^Y$, while the quadratures of 
$W$ are $ R^W = \sqrt{\lambda} {R}^X -\sqrt{1-\lambda} {R}^Y$.  Our main technical result is a proof that no matter what product state is 
prepared on $X$ and $Y$, the beam splitter always increases entropy: $S(Z)\geq \lambda S(X)+(1-\lambda)S(Y)$.  For $\lambda = 1/2$, we prove the
stronger constraint, Eq.(\ref{Eq:ShannonQEPI}).  These fundamental inequalities are the natural quantum generalization of the 
two classically equivalent entropy power inequalities Eq.(\ref{Eq:ShannonEPI}) and Eq.(\ref{Eq:LiebEPI}), and lead to strong new 
upper bounds on the classical communication capacity of additive bosonic channels.}
\label{Fig:BS}
\end{figure}

\section{Applications to classical capacity}

\begin{figure}[htbp]
\includegraphics[width=3in]{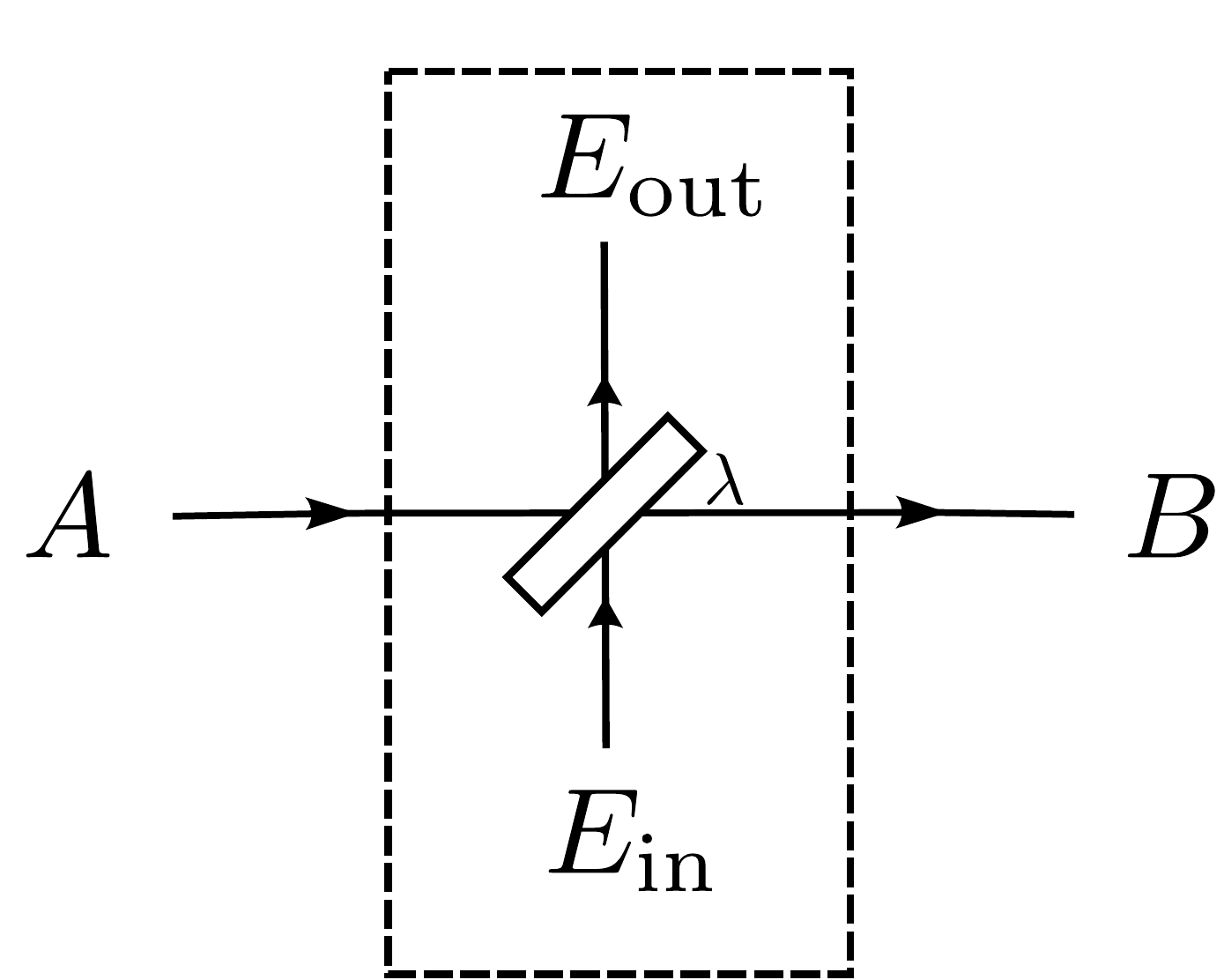}
\caption{An additive noise channel arises when an input signal $A$ interacts via beam splitter with initial environment $E_{\rm in}$ 
followed by a partial trace over  $E_{\rm out}$ resulting in an output signal $B$.  In general, the state of $E_{\rm in}$
can be arbitary.  If $E_{\rm in}$ is in a thermal state with average photon number $N_E$ and the transmissivity of the beam splitter
is $\lambda$, we say we have a thermal noise channel, $\cE_{\lambda,N_E}$.  Letting $\lambda \rightarrow 1$ and $N_{E}\rightarrow\infty$, 
while holding $(1-\lambda) N_E = \nu$ gives is the classical noise channel, which acts 
as $\cE_{\nu}(\rho) = \frac{1}{8\pi \nu}\int d^2\xi W_\xi \rho W_\xi^\dagger e^{-\frac{\xi^T\xi}{8\nu}}$. }
\label{Fig:ThermalDefinition}
\end{figure}

\begin{figure}

\subfigure[]{
\includegraphics[width=3in]{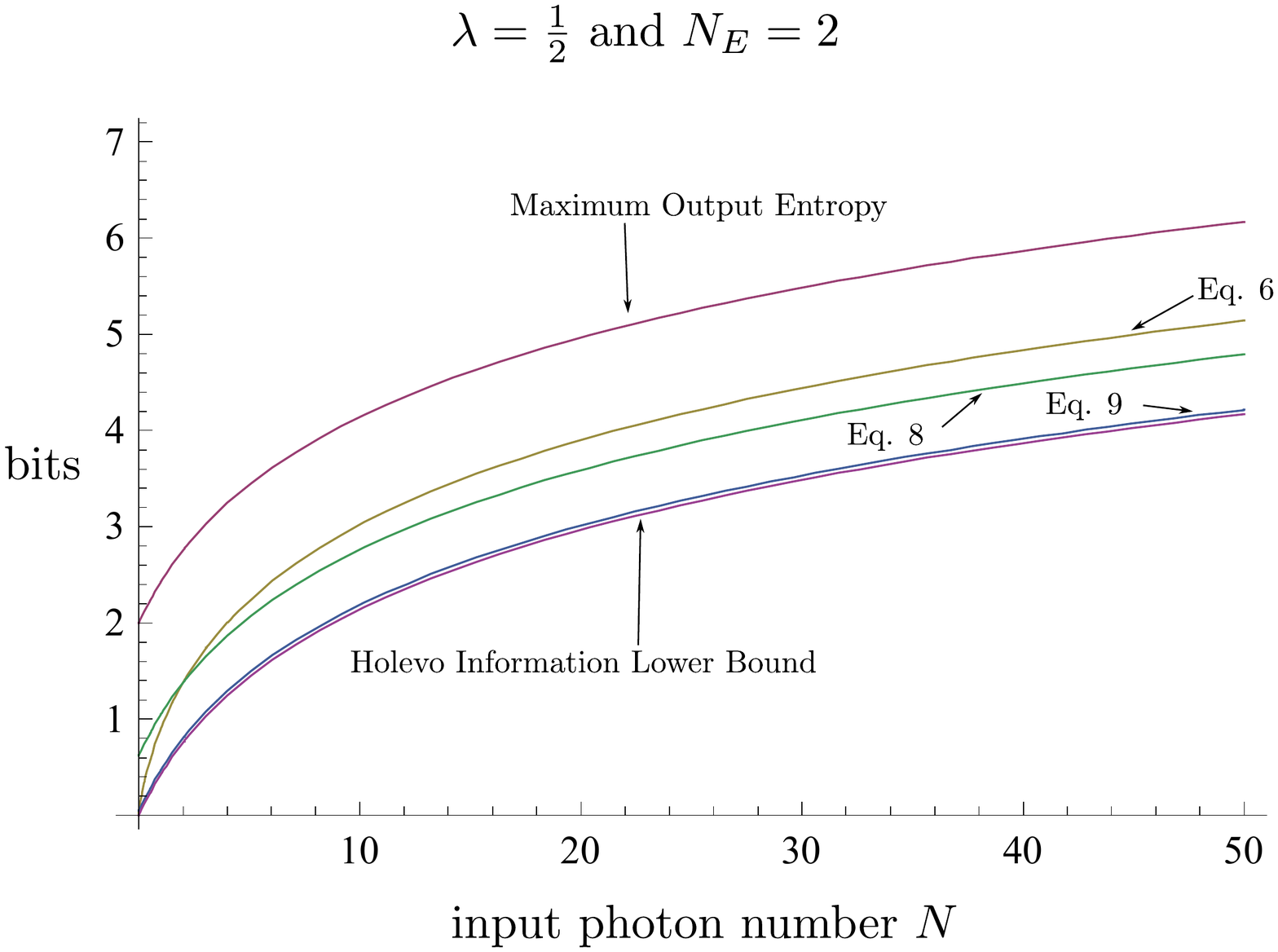}
}

\subfigure[]{
\includegraphics[width=3in]{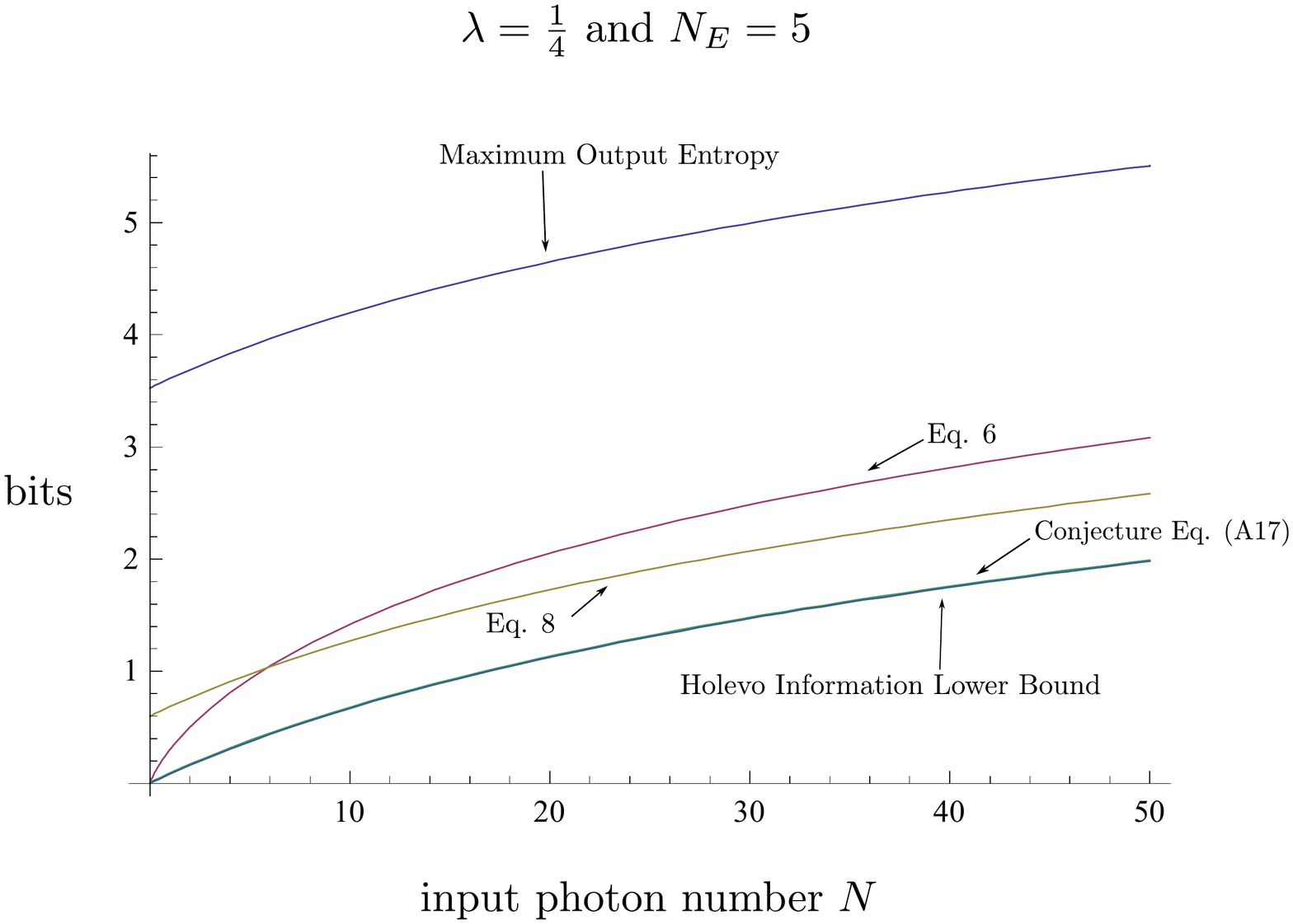}
}
\subfigure[]{
\includegraphics[width=3in]{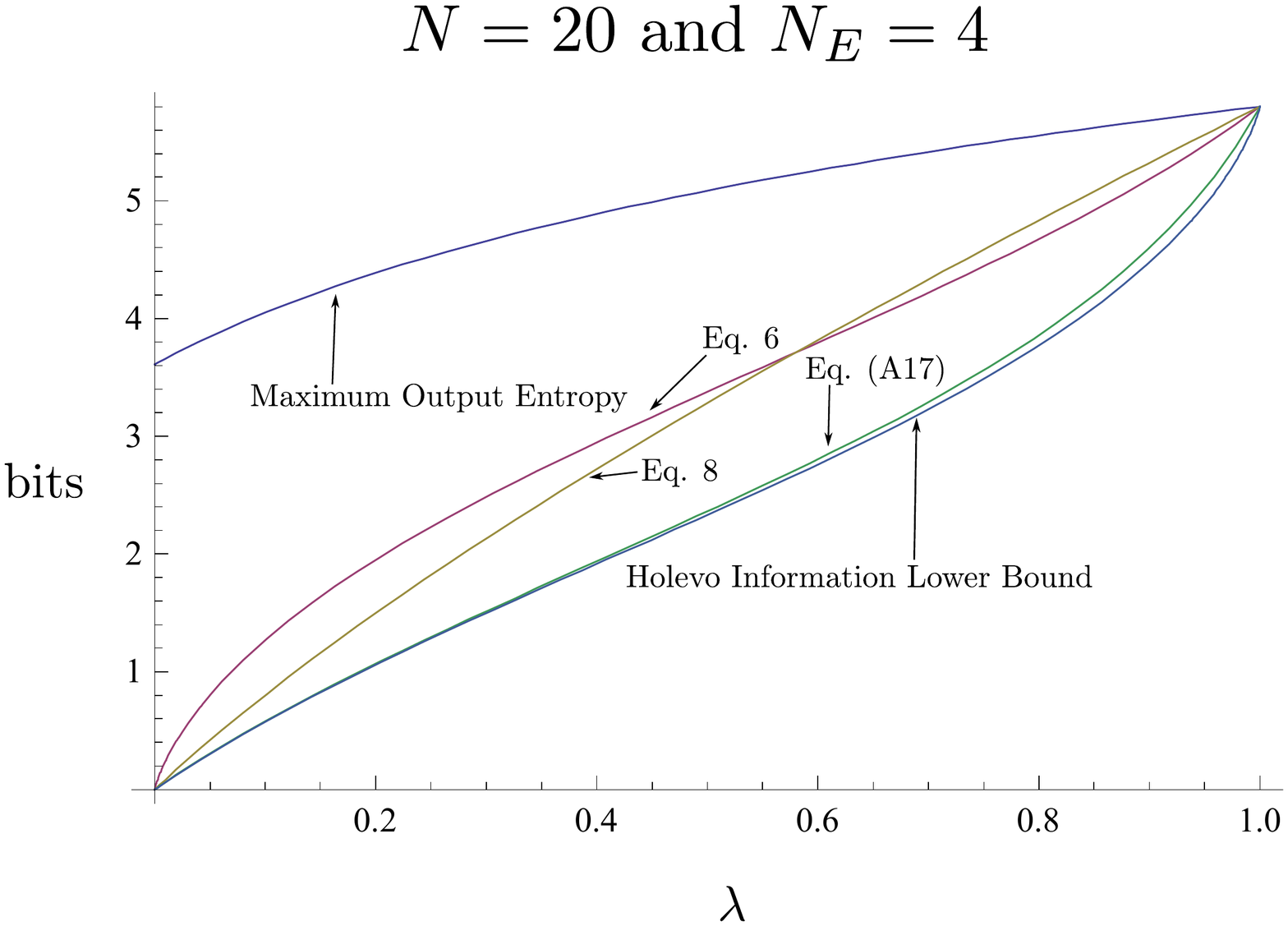}
}

\caption{Known bounds on the classical capacity of thermal noise channels.    Our results show 
that for a general additive noise channel $\cE$ with transmissivity  $\lambda$ 
the capacity with mean signal photon number $N$ satisfies $C(\cE,N) \leq S_{\rm max}(\cE,N) - (1-\lambda)S(E_{\rm in})$ (Eq.~(\ref{Eq:LiebUB})), 
while for general 
thermal noise channel $C(\cE_{\lambda, N_E},N) \leq \left[g(\lambda N+ (1-\lambda)N_E) - (1-\lambda)g(N_E)\right] \frac{1}{\ln 2}$ (Eq.~(\ref{Eq:AE})).  
For a thermal noise channel 
with $\lambda = 1/2$ we 
have $C(\cE_{1/2,N_E})  \leq \left[g\left(\frac{1}{2}(N +N_E)\right) -\ln \left( 1 + e^{g(N_E)}\right)\right] \frac{1}{\ln 2} +  1$, which is
Eq.~(\ref{Eq:ShannonUB}).  
Our conjecture Eq.~(\ref{Eq:ConjPower}) would 
imply $C(\cE_{\lambda, N_E},N) \leq \left[ g(\lambda N+ (1-\lambda)N_E) - \ln(\lambda + (1-\lambda)e^{g(N_E)})\right] \frac{1}{\ln 2}$ for 
general thermal noise channels. 
We compare these bounds to 
the only other known bounds for the capacity of these channels:  the output entropy upper bound and the 
Holevo Information lower bound (both can be found in \cite{HW01}). In (a) we plot these for $\lambda = \frac{1}{2}$ and $N_E = 2$, 
while (b) shows $\lambda = \frac{1}{4}$ and $N_E = 5$.  On the scale of plot (b), 
the conjectured bound Eq.(A17) is indistinguishable from the Holevo information lower bound.  Plot (c) shows capacity bounds for fixed environment
photon number and signal photon number, as a function of channel transmissivity.  }
\label{Fig:Plots}

\end{figure}

Before proving our entropy power inequalities, we consider their
implications for the classical capacity of a thermal noise channel with average thermal 
photon number $N_E$ and transmissivity $\lambda$, $\cE_{\lambda,N_E}$ ( Fig. \ref{Fig:ThermalDefinition}).  Eq.~(\ref{Eq:Holevo}) is the 
best known achievable rate for classical communication over this channel with  average signal photon number $N$ \cite{HW01}.  
In general the capacity exceeds the Holevo information, 
which corresponds to an enhanced communication capability from entangled signal states \cite{Hastings09}.  However, 
for the pure loss channel, $\cE_{\lambda,0}$, the bound is tight, 
giving capacity $C(\cE_{\lambda,0}) = g(\lambda N)/\ln 2$ \cite{GGLMSY04}.  Using the method of 
additive extensions \cite{AE08} gives the upper bound
\begin{align}\label{Eq:AE}
C(\cE_{\lambda,N_E},N)\leq g\left(\frac{\lambda N}{(1-\lambda)N_E +1}\right)\frac{1}{\ln 2}.
\end{align}
While Eq.~(\ref{Eq:AE}) follows from an elementary argument (see Appendix), as far as we know it is new.

Closely related to capacity, the minimum output entropy is a measure of a channel's noisiness  \cite{KR01,Shor04}. 
 Indeed, the classical capacity of any channel $\cE$ satisfies
 \begin{align}
 C(\cE,N) \leq S_{\rm max}(\cE,N) - \lim_{n \rightarrow \infty}\frac{1}{n}S_{\rm min}(\cE^{\otimes n}), 
 \end{align}
where $S_{\rm max}(\cE,N) = \max_{\Tr[H\rho]\leq 2N+1}S(\cE(\rho))$ is the maximum output entropy with photon number constraint $N$ 
($H = \frac{1}{2}[P^2+Q^2]$ is the harmonic oscillator Hamiltonian and $(H-1)/2$ is the number operator), $S_{\rm min}(\cE) = \min_{\rho}S(\cE(\rho))$
is the minimum output entropy, and $\cE^{\otimes n}$ is the $n$-fold tensor product representing $n$ parallel uses of the channel.  
The difficulty in applying this upper bound is the infinite limit in the second term, which prevents us from evaluating the right
hand side.  However, for additive noise channels our EPIs give lower bounds\footnote{In jargon,
we prove an additive lower bound on $S_{\rm min}(\cE^{\otimes n})$ giving a single-letter expression. } on 
$S_{\rm min}(\cE^{\otimes n})$, allowing simple upper bounds on the capacity.  In particular, from Eq.~(\ref{Eq:LiebQEPI}) 
we find the capacity of a thermal noise channel with environment photon 
number $N_E$ and signal photon number $N$ satisfies
\begin{align}\label{Eq:LiebUB}
C(\cE_{\lambda,N_E},N) \leq \left[ g\left(\lambda N + (1-\lambda)N_E\right) -(1-\lambda)g(N_E)\right] \frac{1}{\ln 2},
\end{align}
while for $\lambda = \frac{1}{2}$ Eq.~(\ref{Eq:ShannonQEPI}) implies the stronger
\begin{align}\label{Eq:ShannonUB}
C(\cE_{1/2,N_E},N) & \leq \left[g\left(\frac{1}{2}(N +N_E)\right) -\ln \left( 1 + e^{g(N_E)}\right)\right]\frac{1}{\ln 2} +  1.
\end{align}
This bound differs from the Holevo lower bound by no more than $0.06$ bits (Fig. \ref{Fig:Plots}).

\section{Divergence-based Quantum Fisher Information}

Fisher information is a key tool in the proof of the classical EPI\cite{Stam59,Blachman65}, however there is no unique quantum 
Fisher information \cite{Petz02}.  We introduce a particular quantum Fisher information defined in terms of the quantum divergence,
$S(\rho\|\sigma) = \Tr \rho ({\log_2} \rho - {\log_2} \sigma)$. 
Given a smooth family of states $\rho_{\theta}$, we define the divergence-based quantum Fisher information as the second derivative of divergence along the path:
\begin{align}
J(\rho_{\theta}; \theta )|_{\theta = \theta_0} = \partial_{\theta}^2 S(\rho_{\theta_0}\| \rho_{\theta})|_{\theta= \theta_0}.
\end{align}
This is nonnegative ($J(\rho_\theta;\theta) \geq 0$), 
additive ($J(\rho_{\theta}^A \otimes \rho_{\theta}^B; \theta) = J(\rho_{\theta}^A; \theta)+J(\rho_{\theta}^B; \theta)$) and satisfies
data processing ($J(\cE(\rho_{\theta}); \theta) \leq J(\rho_\theta; \theta)$ for any physical map $\cE$).  It also satisfies the
reparametrization formulas, $J(\rho_{c\theta}; \theta)|_{\theta=0} = c^2J(\rho_{\theta}; \theta)|_{\theta = 0}$ 
and $J(\rho_{\theta+c};\theta)|_{\theta=0} = J(\rho_{\theta};\theta)|_{\theta=c}$  \cite{KS12}.

\section{Quantum Diffusion}
Fisher information appears in the classical EPI proof because of its 
relation to the entropy production rate under the addition of gaussian noise via the de Bruijin identity, 

\begin{align}
\frac{dH(X+\sqrt{t}Z)}{dt}|_{t=0} = \frac{1}{2}J(X).
\end{align}
Here $X$ is an arbitrary variable, $Z$ is an independent normal variable with unit variance and $J(X)$ is the classical Fisher information
of the ensemble $\{X+\theta\}_{\theta \in \mathbb{R}}$.  The variable $X+\sqrt{t}Z$ arises from a diffusion with
initial state $X$ running for time $t$.

To explain our quantum de Bruijin identity, we must first discuss quantum diffusion processes.  A quantum Markov process is associated
with a Liouvillean  $ \cL(\rho)$ and governed by a Markovian master equation,

\begin{align}
\frac{d\rho}{dt} = \cL(\rho).
\end{align}
Our process of interest has $\cL(\rho) = -\frac{1}{4}\sum_{i}[R_i,[R_i,\rho]]$, and corresponds to adding 
gaussian noise in phase space \cite{Hall00} (see Fig. \ref{Fig:BS} for definitions of the quadratures $R_i$). 
 We denote the action of running this process
for time $t$ on initial state $\rho_0$ by $e^{\cL t}(\rho_0)$.  We 
want to relate the entropy production rate of our quantum diffusion to a Fisher information, 
but what ensemble should we use?  We choose $2n$ separate ensembles of states, 

\begin{align}
\rho_{\theta}^{R_i} & = D_{R_i}(\theta)\rho_0 D_{R_i}^\dagger(\theta) \label{Eq:EnsembleDefinition}
\end{align}
where $D_{R_i}$ is a displacement operator along the $R_i$ axis in phase space.
We then find for sufficiently smooth $\rho_0$ that

\begin{align}\label{Eq:QuantumDeBruijin}
\frac{dS(e^{t\cL}(\rho_0))}{dt}|_{t=0} = \sum_{i=1}^{2n}J(\rho_{\theta}^{R_i}; \theta) =: \tilde{J}(\rho_0).
\end{align}

The smoothness requirements are necessary because in infinite dimensions moving a derivative inside the trace of a function 
(easily justified by linearity in finite dimensions) is only possible if the function is sufficiently smooth.
To avoid
excessive technicalities and focus on the main thrust of our arguments we 
simply assume the required smoothness.  This entails little loss of generality, since the entropy on states of bounded energy 
is continuous \cite{Wehrl78}, so 
one can hope to approximate non-smooth functions with smooth ones to obtain the desired result.  
Because of this, such smoothness requirements are rarely considered in proofs of the classical EPI 
\cite{Stam59,Blachman65,rioultwo,VG06} or considerations of its quantum counterparts \cite{Giovannettietal04, Giovannetti10}
\footnote{Indeed, while the proof of the classical EPI is generally attributed to 
Stam \cite{Stam59} and Blachman \cite{Blachman65}, a full consideration justifying the interchange of derivatives
and integrals seems to have first been given by Barron in 1984\cite{Barron84}. }.  

\section{Proof of quantum entropy power inequality}

\begin{figure}[htbp]
\includegraphics[width=6in]{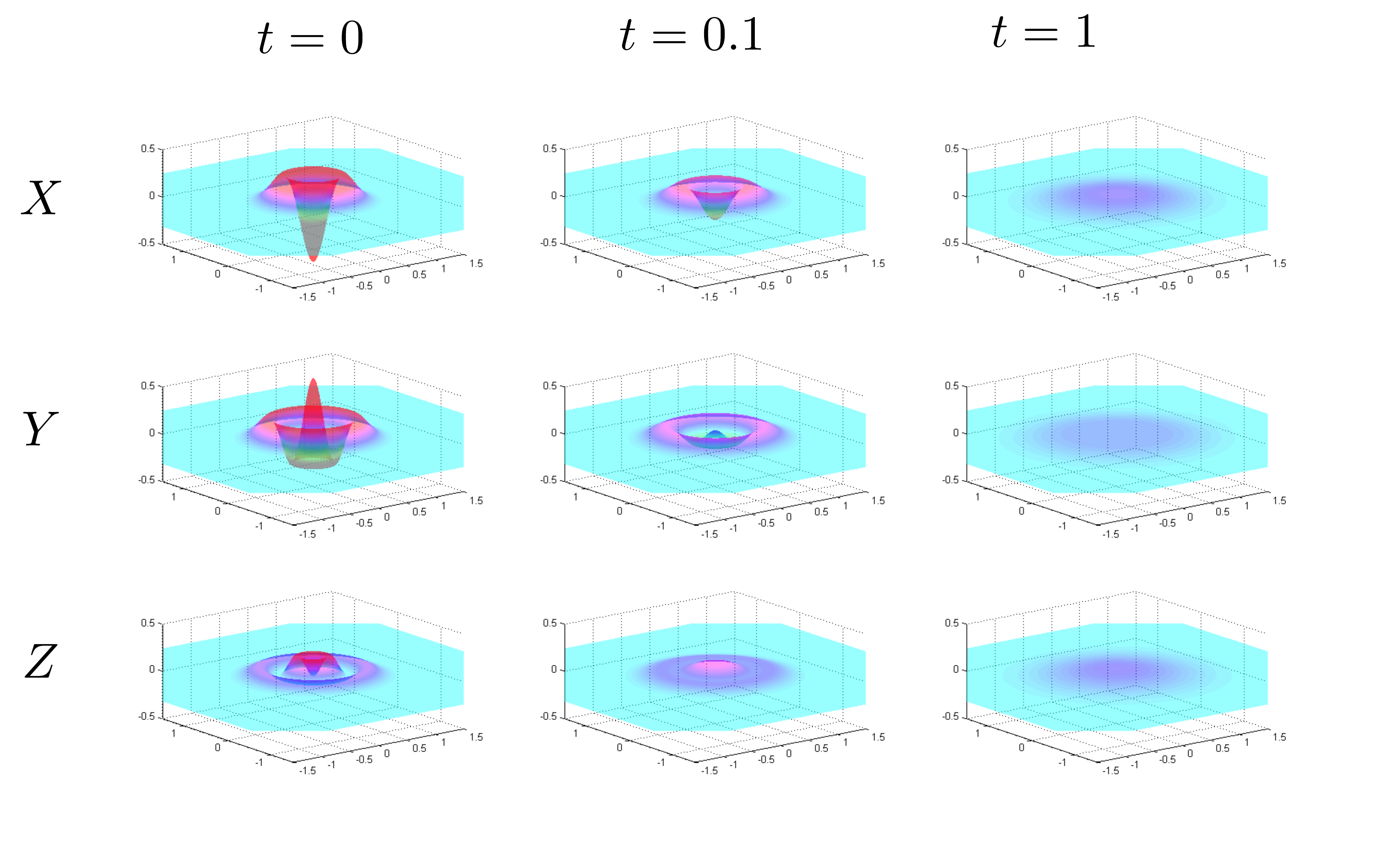}
\caption{Using the evolution of the inputs and output of a beam splitter under diffusion to prove the quantum entropy power inequality.  
The state of a bosonic system can be represented with a Wigner function,
which is the Fourier transform of the state's characteristic function $\chi_{\rho}(\xi) = \Tr \rho W_{\xi}$.  Here the displacement operator
is $W_{\xi} = e^{i\xi^T J R}$ with  $J =\binom{\,\, \,\,\,0 \, \,\,  1\,}{-1\,\,\,  0\,}^{\oplus n}$ and $R = (Q_1,P_1,\dots,Q_n,P_n)$ for
an $n$ mode system. The
Wigner function is a ``quasi-probability'' distribution, in that it integrates to 1 but may be negative. The left column plots, from 
top to bottom, the Wigner functions of an input mode $X$ prepared in a $1$ photon state, an input $Y$ prepared in a $2$ photon state, 
and the resulting output $Z$ when $X$ and $Y$ combine at a $50:50$ beam splitter.  The second column shows the same states when evolved according
to a quantum diffusion for a time $t=0.1$.  The third column shows the states after diffusion for time $t=1$.  In the late time limit, 
all three states approach the same thermal state, and so satisfy Eq.~(\ref{Eq:LiebQEPI}) with equality.  The convexity of Fisher information
can be used, together with the quantum de Bruijin identity, to show that any violation of Eq.~(\ref{Eq:LiebQEPI}) would be amplified under
the diffusion.  Since there is no violation for $t\rightarrow \infty$, there is therefore no violation at $t=0$.  As a result, no matter
what input states are chosen, we have Eq.~(\ref{Eq:LiebQEPI}).}
\label{Fig:Diffusion}
\end{figure}

Our path to the quantum entropy power inequality combines the quantum de Bruijin identity, 
Eq.~(\ref{Eq:QuantumDeBruijin}), with a convexity property of the quantum Fisher information.  In 
particular, we require that the Fisher information of the output of a beam splitter satisfy
\begin{align}\label{Eq:FIConvexity}
\tilde{J}(\rho^{X \boxplus_\lambda Y})\leq \lambda \tilde{J}(\rho^{X}) +(1-\lambda) \tilde{J}(\rho^{Y}).  
\end{align}
The proof of this relation relies on elementary properties of $\tilde{J}$, and follows the analogous classical proof \cite{Zamir98}.

Roughly speaking, Eq.~(\ref{Eq:LiebQEPI}) is proven by subjecting inputs $\rho_X$ and $\rho_Y$  to a 
quantum diffusion for time $t$.  As $t \rightarrow \infty$, both initial states approach a thermal state with average
photon number $(t-1)/2$, as does the combination of the two states at a beam splitter.  Since both inputs, as well as the 
beam splitter's output, approach the same state in the limit, Eq.~(\ref{Eq:LiebQEPI}) is satisfied with equality.  We then use 
Eq.~(\ref{Eq:FIConvexity}) together with the quantum de Bruijin identity to show that any violation of Eq.~(\ref{Eq:LiebQEPI})
would be amplified as $t$ grows.  Since in the limit $t\rightarrow \infty$ the violation is zero, we conclude that no such violation
exists.  This argument also applies to multi-mode systems, so Eq.~(\ref{Eq:LiebQEPI}) is true for these too.

The proof of Eq.~(\ref{Eq:ShannonQEPI}) is similar in spirit to our proof of Eq.~(\ref{Eq:LiebQEPI}) and 
Blachmann's proof \cite{Blachman65} of Eq.~(\ref{Eq:ShannonEPI}).  Rather than convexity, we use a quantum version of Stam's inequality:
\begin{align}
\frac{2}{\tilde{J}(\rho^{X\boxplus_{1/2}Y})} \geq \frac{1}{\tilde{J}(\rho^{X})}+\frac{1}{\tilde{J}(\rho^{Y})},
\end{align}  
and consider a ratio rather than a difference (see Methods section or \cite{KS12} for more details).

\section{Discussion/outlook}

Some authors have hoped the lower bound of Eq.~(\ref{Eq:Holevo}) is equal to the capacity \cite{Giovannettietal04,GSE08}; 
There is evidence both for \cite{Giovannettietal04,SEW05,GSE08,Guha08,Giovannetti10} and against \cite{Hastings09,SSY11} this conjecture.  
It has been related to an ``entropy photon-number inequality'' which, if true, 
would imply this equivalence,  but despite concerted effort no proof has been found.  Our quantum EPIs  
more closely resemble the classical inequalities  than does the proposed inequality of \cite{Giovannettietal04}, 
allowing us to often rely on classical proof strategies.  

We expect our results to find a variety of applications to bosonic systems.  The analysis of classical 
network models like broadcast \cite{Bergmans74} 
and interference channels \cite{Costa85a,Costa85b}
relies on EPIs, so network quantum information theory is a good place to start \cite{YHD11,FHSSW11,GSW11}.  
Quantum EPIs may also find applications in the development of noncommutative central limit theorems \cite{Hudson73,Barron86,Guha08}. 

There are many potential generalizations for our inequalities.  For example, one could follow Costa \cite{Costa85a} and 
show that $\exp\left[ \frac{1}{n}S(e^{t\cL}(\rho))\right]$ is concave as a function of $t$.  Foremost, however, is proving the 
analogue of Eq.~(\ref{Eq:ShannonQEPI}) for $\lambda \neq \frac{1}{2}$.  One would hope that

\begin{align}\label{Eq:ConjPower}
e^{\frac{1}{n}S(X\boxplus_\lambda Y)} \geq \lambda e^{\frac{1}{n}S(X)} + (1-\lambda)e^{\frac{1}{n}S(Y)},
\end{align}    

but we have not yet found a proof.  Such a result would give bounds on the capacity of the thermal and classical noise channels
to within $0.16$ bits, answering the capacity question for all practical purposes.

\section{Methods}

{\bf Details of proof of Eq.(\ref{Eq:LiebQEPI})}

We would like to show that, given input states $\rho_X$ and $\rho_Y$, 
\begin{align}\label{Eq:ExplicitQEPI}
S\left(\cB_\lambda (\rho_X \otimes \rho_Y)\right) \geq \lambda S\left(\rho_X\right)+ (1-\lambda) S\left(\rho_Y\right),
\end{align} 
where $\cB_\lambda (\rho_X \otimes \rho_Y)$ denotes the map from inputs to outputs of a beam splitter
with transmissivity $\lambda$.  To do so, we let 
\begin{align}
s(t)&=S\left(e^{t\cL}\left(\cB_\lambda (\rho_X \otimes \rho_Y)\right)\right)
 - \lambda S\left(e^{t\cL}\left(\rho_X\right)\right) - (1-\lambda) S\left(e^{t\cL}\left(\rho_Y\right)\right) 
\end{align}
be the difference between the two sides of the desired inequality.  Since as $t\rightarrow \infty$
all states involved approach a gaussian state with photon number $(t-1)/2$, one expects that 
$\lim_{t \rightarrow \infty}s(t) = 0$, and indeed this is the case \cite{KS12}.  Furthermore,
using the quantum de Bruijin identity to differentiate, we find
\begin{align}
s^{\prime}(t) = \tilde{J}\left(e^{t\cL}\left(\cB_\lambda (\rho_X \otimes \rho_Y)\right)\right) 
-\lambda \tilde{J}\left(e^{t\cL}\left(\rho_X\right)\right) -(1-\lambda) \tilde{J}\left(e^{t\cL}\left(\rho_Y\right)\right).
\end{align}
Finally, using the fact that 
\begin{align}
e^{t\cL}\left(\cB_\lambda (\rho_X \otimes \rho_Y)\right) = \cB_\lambda \left(e^{t\cL}(\rho_X) \otimes e^{t\cL}(\rho_Y)\right)
\end{align}
we find that 
\begin{align}
s^{\prime}(t) = \tilde{J}\left(\cB_\lambda \left(e^{t\cL}(\rho_X) \otimes e^{t\cL}(\rho_Y)\right)\right) 
-\lambda \tilde{J}\left(e^{t\cL}\left(\rho_X\right)\right) -(1-\lambda) \tilde{J}\left(e^{t\cL}\left(\rho_Y\right)\right)
\end{align}
so that by Eq.~(\ref{Eq:FIConvexity}), we have $s^\prime(t) \leq 0$.  Since $\lim_{t\rightarrow \infty}s(t)= 0$ \cite{KS12} and $s(t)$ is
monotonically decreasing, we thus find that $s(0)\geq 0$.  In other words, we get Eq.~(\ref{Eq:ExplicitQEPI}).

{\bf Proof sketch of Eq.(\ref{Eq:ShannonQEPI})}

As mentioned above, to establish Eq.~(\ref{Eq:ShannonQEPI}), rather than using convexity, we appeal to a quantum version of Stam's inequality:
\begin{align}\label{Eq:QuantumStam}
\frac{2}{\tilde{J}(\rho^{X\boxplus_{1/2}Y})} \geq \frac{1}{\tilde{J}(\rho^{X})}+\frac{1}{\tilde{J}(\rho^{Y})},
\end{align}  
whose proof along the lines of \cite{Zamir98} can be found in \cite{KS12}.  In fact, we let $X$ evolve according 
to a quantum diffusion for time $F(t)$ and $Y$ evolve for $G(t)$ with $\lim_{t\rightarrow \infty}F(t) = \lim_{t\rightarrow \infty}G(t) = \infty$.
Then, letting $\rho^{F(t)}_X = e^{F(t)\cL}(\rho_X)$, $\rho^{G(t)}_Y = e^{G(t)\cL}(\rho_Y)$ and 
$\rho^{\frac{1}{2}[F(t)+G(t)]}_{Z} = e^{\frac{1}{2}[F(t)+G(t)]\cL}(\rho_{X\boxplus_{1/2}Y})$, 
we can show that as $t\rightarrow \infty$ the ratio, 

\begin{align}
h(t)= \frac{\frac{1}{2}\exp\left(\frac{1}{n}S\left(\rho^{F(t)}_X\right)\right)+\frac{1}{2}\exp\left(\frac{1}{n}S\left(\rho^{G(t)}_Y\right)\right)}{\exp\left( \frac{1}{n}S\left(\rho^{\frac{1}{2}[F(t)+G(t)]}_{Z}\right)\right)}
\end{align}
approaches $1$.  Using the quantum de Bruijin identity to evaluate $h^\prime(t)$ allows us to find a differential equation 
for $F$ and $G$ that ensures, together with Eq.~(\ref{Eq:QuantumStam}), $h^\prime(t)\geq 0$.  This allows us to conclude that $h(0)\leq 1$, 
which implies Eq.~(\ref{Eq:ShannonQEPI}).

\section*{Acknowledgments}
We are grateful to  Charlie Bennett, Jay Gambetta, and John Smolin for helpful comments and advice,  
Saikat Guha for discussions of the entropy photon number inequality, and Mark Wilde for comments and suggesting references.  
We were both supported by DARPA QUEST program under contract no.HR0011-09-C-0047.

\appendix

\section{Implications for classical capacities}

\begin{lemma}
Measured in nats, the classical capacity of the single-mode thermal noise channel, $\cE_{\lambda,N_E}$, with average photon number $N$ satisfies
\begin{align}
C(\cE_{\lambda,N_E},N) & \leq g\left(\lambda N + (1-\lambda)N_E\right) -(1-\lambda)g(N_E).  
\end{align}
If $\lambda = \frac{1}{2}$, we also have
\begin{align}
C(\cE_{1/2,N_E},N) & \leq g\left(\frac{1}{2}(N +N_E)\right) -\ln \left( 1 + e^{g(N_E)}\right) +  \ln 2.
\end{align}
The capacity in bits can be obtained by dividing the formla for nats by $\ln 2$.
\end{lemma}

\begin{proof}
We begin with the Holevo-Schumacher-Westmoreland formula for the classical capacity \cite{SW97,Holevo98} (see also \cite{YO93}), 
from which we conclude
\begin{align}
C(\cE_{\lambda,N_E},N) & = \lim_{n\rightarrow \infty} \frac{1}{n}\max_{\{p_x,\phi_x^n \Tr[\sum_x p_x \phi_x H_{(n)}]\leq (2N+1)n\}}
\left[S\left(\cE^{\otimes n}_{\lambda,N_E}\left(\sum_x p_x \phi_x^n\right) \right) - \sum_x p_x S\left(\cE^{\otimes n}_{\lambda,N_E}(\phi_x^n) \right)\right]\\
& \leq \lim_{n\rightarrow \infty} \frac{1}{n}\max_{\Tr[\phi_n H_{(n)}]\leq n(2N+1)  } \left[S\left(\cE^{\otimes n}_{\lambda,N_E}\left(\phi_n\right) \right)\right] - 
 \lim_{n\rightarrow \infty} \frac{1}{n}\min_{\phi_n}\left[S\left(\cE^{\otimes n}_{\lambda,N_E}\left(\phi_n\right) \right)\right]\\
& \leq \max_{\Tr[\phi H] \leq (2N+1)} \left[S\left(\cE_{\lambda,N_E}\left(\phi\right) \right)\right]
-\lim_{n\rightarrow \infty} \frac{1}{n}\min_{\phi_n}\left[S\left(\cE^{\otimes n}_{\lambda,N_E}\left(\phi_n\right) \right)\right]\\
& = g\left(\lambda N + (1-\lambda)N_E\right)-\lim_{n\rightarrow \infty} \frac{1}{n}\min_{\phi_n}\left[S\left(\cE^{\otimes n}_{\lambda,N_E}\left(\phi_n\right) \right)\right], \label{Eq:PartialUB}
\end{align}
where the first inequality is elementary, the second is due to subadditivity of entropy, and the final equality comes from the 
fact that gaussian states maximize entropy for any given power level \cite{Wolfetal06}.  We complete the prove by using 
Eq.~(\ref{Eq:LiebQEPI}) to show
\begin{align}
S\left(\cE^{\otimes n}_{\lambda,N_E}\left(\phi_n\right) \right) & \geq \lambda S(\phi_n) + (1-\lambda)S(E_{\rm in}^n)\\
& = \lambda S(\phi_n) +(1-\lambda)nS(E_{\rm in}) 
& \geq (1-\lambda)nS(E_{\rm in}) = n(1-\lambda)g(N_{E}), 
\end{align}
from which we conclude with Eq.~(\ref{Eq:PartialUB}) that 
\begin{align}
C(\cE_{\lambda,N_E}) & \leq g\left(\lambda N + (1-\lambda)N_E\right) -(1-\lambda)g(N_E).
\end{align} 

For $\lambda = \frac{1}{2}$, we also have
\begin{align}
e^{\frac{1}{n}S\left(\cE^{\otimes n}_{1/2,N_E}(\phi_n)\right)}& \geq \frac{1}{2}e^{\frac{1}{n}S(\phi_n)} + \frac{1}{2}e^{\frac{1}{n}S(E_{\rm in}^n)}
\end{align}
from Eq.~(\ref{Eq:ShannonQEPI}), which gives us 
\begin{align}
\frac{1}{n}S\left(\cE^{\otimes n}_{1/2,N_E}(\phi_n)\right) &\geq \ln \left[ 1 + e^{\frac{1}{n}S(E_{\rm in}^n)} \right] - \ln 2, 
\end{align}
so that 
\begin{align}
\frac{1}{n}S\left(\cE^{\otimes n}_{1/2,N_E}(\phi_n)\right) & \geq \ln\left[ 1 + e^{g(N_E)} \right] - \ln 2.
\end{align}
This allows us to conclude, together with Eq.~(\ref{Eq:PartialUB})

\begin{align}
C(\cE_{1/2,N_E},N) & \leq g\left(\frac{1}{2} \left(N +N_E \right)\right) - \ln\left[ 1 + e^{g(N_E)} \right] +\ln 2.
\end{align}
\end{proof}

In a similar fashion, we can prove that Eq.~(\ref{Eq:ConjPower}) would imply that 

\begin{align}\label{Eq:ConjThermalUB}
C(\cE_{\lambda,N_E}) & \leq g\left( \lambda N +(1-\lambda)N_E \right) - \ln\left[ \lambda + (1-\lambda)e^{g(N_E)} \right].
\end{align}

We also prove the following lemma, which shows an upper bound for the classical noise channel is implied by Eq.~(\ref{Eq:ConjPower}).

\begin{lemma}
If we have
\begin{align}
e^{\frac{1}{n}S(X\boxplus_\lambda Y)} \geq \lambda e^{\frac{1}{n}S(X)} + (1-\lambda)e^{\frac{1}{n}S(Y)},
\end{align}
then the classical capacity of the classical noise channel measured in nats satisfies 
\begin{align}
g(N+\nu)- g(\nu) \leq C(\cE_{\nu},N) \leq g(N+\nu)- \ln(1+e\nu).
\end{align}
The upper and lower bounds differ by no more than $0.11$ nats or $0.16$ bits.
\end{lemma}

\begin{proof}
The lower bound comes from \cite{HW01}, so we need only establish the upper bound.  This is done by evaluating the bound in 
Eq.~(\ref{Eq:ConjThermalUB}), 
\begin{align}
C(\cE_{\lambda,N_E},N) & \leq g\left( \lambda N +(1-\lambda)N_E \right) - \ln\left[ \lambda + (1-\lambda)e^{g(N_E)} \right],
\end{align}
in the limit $\lambda \rightarrow 1$, $N_E \rightarrow \infty$ with $(1-\lambda)N_E = \nu$.  Since

\begin{align}
(1-\lambda)e^{g(N_E)} & = (1-\lambda)\frac{\left(N_E+1\right)^{N_E+1}}{N_E^{N_E}}\\
& = (1-\lambda)(N_E+1)\left(\frac{N_E+1}{N_E}\right)^{N_E}\\
 &= \nu e,
\end{align}

we find 
\begin{align}
C(\cE_{\nu} ,N) \leq g(N+\nu)- \ln(1+e\nu).
\end{align}

That the gap between upper and lower bounds is no more than $0.11$ can easily be established by calculus.

\end{proof}

\begin{lemma}

Measured in nats, the classical capacity of the single-mode thermal noise channel, $\cE_{\lambda,N_E}$, with average photon number $N$ satisfies
\begin{align}
C(\cE_{\lambda,N_E}) \leq \left[g\left(\frac{\lambda N}{(1-\lambda)N_E +1}\right)\right].
\end{align}
\end{lemma}

\begin{proof}
First, we let $\cA_{G}$ be a pure-gain quantum channel with gain $G$, mapping covariance matrix $\gamma\rightarrow G\gamma + (G-1)I$.
Then, note that $\cE_{\lambda,N_E} = \cA_{G} \circ \cE_{a,0}$
with $G = (1-\lambda)N_E+1$, $a  = \frac{\lambda}{(1-\lambda)N_E+1}$.  Since  the capacity of $\cE_{a,0}$ is known to be $g(aN)/\ln 2$ \cite{GGLMSY04}, 
we have
\begin{align}
C(\cE_{\lambda,N_E},N)  = C\left(\cA_{G} \circ \cE_{a,0},N\right) \leq C(\cE_{a,0},N) =  \left[g\left(\frac{\lambda N}{(1-\lambda)N_E +1}\right)\right]\frac{1}{\ln 2}.
\end{align}

\end{proof}

\bibliographystyle{apsrev}
%\bibliography{shortq}

\end{document}